\documentclass[10pt, conference]{IEEEtran}
\IEEEoverridecommandlockouts
\usepackage{pifont}
\usepackage{tabu}
\usepackage{supertabular,booktabs}

\usepackage{longtable}

\usepackage[hyphens]{url}
\usepackage{multicol}
\usepackage{amsmath, amssymb, bm, cite, epsfig, psfrag, mathtools}
\usepackage{epstopdf}
\usepackage{changepage}
\usepackage{graphicx}
\usepackage{fancyhdr}
\usepackage{bbm}
\newcommand*{\Scale}[2][4]{\scalebox{#1}{$#2$}}%
\usepackage{algorithm,algorithmic}
\usepackage{array}
\usepackage[margin=10pt,font=small]{caption}
\usepackage{bbm}
\usepackage{multirow}
\usepackage[usenames,dvipsnames]{xcolor}

\usepackage{subfigure}
\usepackage{etoolbox}
\usepackage{pbox}
\usepackage[width=0.48\textwidth]{caption}
\usepackage{colortbl}
\usepackage{bm}
\usepackage{caption}
\usepackage{amsmath}
\usepackage{color}

\usepackage{enumitem}
\usepackage{datetime}
\usepackage{amsthm}
\usepackage{dsfont}
\usepackage{mathtools}
\usepackage[T1]{fontenc}

\newtoggle{conference}
\togglefalse{conference} 
\usepackage[letterpaper,left=0.75in,right=0.75in,top=0.75in,bottom=0.75in,footskip=.25in]{geometry}
\graphicspath{{figures/}}

\def\beq{\begin{equation}}
\def\eeq{\end{equation}}
\def\beqa{\begin{eqnarray}}
\def\eeqa{\end{eqnarray}}
\def\beqan{\begin{eqnarray*}}
\def\eeqan{\end{eqnarray*}}

\def\EE{{\mathbb{E}}}
\def\PP{{\mathbb{P}}}

\def\argmin{\mathop{\mathrm{arg\,min}}}

\newcommand{\Fc}{{\cal F}}

\newtheorem{definition}{Definition}
\newtheorem{theorem}{Theorem}

\newtheorem{example}{Example}

\def\FF{{\mathbb{F}}}

\def\tm1{t\! - \! 1}
\def\tp1{t\! + \! 1}

\def\dbf{\mathbf{d}}

\def\fbf{\mathbf{f}}

\def\pbf{\mathbf{p}}

\def\qbf{\mathbf{q}}

\def\Uc{\mathcal{U}}

\def\fsf{ {\sf f}}

\begin{document}

\bibliographystyle{IEEEtran}


\title{\Huge Correlation-Aware Distributed Caching and Coded Delivery}

\author{P. Hassanzadeh, A. Tulino, J. Llorca, E. Erkip
\thanks{P. Hassanzadeh  and  E. Erkip are with the ECE Department of New York University, Brooklyn, NY. Email: \{ph990, elza\}@nyu.edu}
\thanks{J. Llorca  and A. Tulino are with Bell Labs, Nokia, Holmdel, NJ, USA. Email:  \{jaime.llorca, a.tulino\}@nokia.com}
\thanks{A. Tulino is with the DIETI, University of Naples Federico II, Italy. Email:  \{antoniamaria.tulino\}@unina.it}
}

\maketitle

\begin{abstract}
Cache-aided coded multicast 
leverages side information at wireless edge caches to efficiently serve multiple groupcast demands via common multicast transmissions, 
leading to load reductions that are proportional to the aggregate cache size.
However, the increasingly unpredictable and personalized nature of the content that users consume challenges the efficiency of existing caching-based solutions in which only {\em exact} content reuse is explored.
This paper generalizes the cache-aided coded multicast problem to a source 
compression with distributed side information problem that specifically accounts for the correlation among the content files.
It is shown how joint file compression 
during the caching and delivery phases 
can provide load reductions that go beyond those achieved with existing schemes. 
This is accomplished through a lower bound on the fundamental rate-memory trade-off as well as a correlation-aware achievable scheme, shown to significantly outperform state-of-the-art correlation-unaware solutions, while approaching the limiting rate-memory trade-off.
\end{abstract}

\section{Introduction and Setup}~\label{sec:Introduction}

We consider a broadcast caching network with one sender (e.g., base station) connected to $n$ receivers (e.g., access points or user devices) $\mathcal U=\{1,\dots,n\}$ via a shared error-free multicast link.
The sender has access to a file library $\mathcal F=\{1,\dots,m\}$ composed of $m$ files. Each file $f\in\mathcal F$ is represented by a vector of independent and identically distributed (i.i.d.) binary symbols of length $F$, ${\sf W}_f\in\mathbb F_2^F$, and therefore $H({\sf W}_f)=F$ \footnote{The results derived in this paper can be easily extended to the case in which the symbols within each file are not necessarily binary and independent.}. 
Each receiver has a cache of size $MF$ bits for a real number $M\in[0,m]$, and receivers request files in an 
i.i.d. manner according to a demand distribution $\qbf = (q_1,\dots,q_m)$,
where $q_f$ denotes the probability of requesting file $f\in\mathcal F$. 
The file library realization is denoted by $\{W_f  :f\in\mathcal F\}$, and we further assume that files can be correlated, 
i.e., $H({\sf W}_{f_i}|{\sf W}_{f_j})\leq F$, $\forall f_i,f_j\in\mathcal F$, and hence 
$H({\sf W}_1,\dots,{\sf W}_m)\leq mF$.
Such correlations are especially relevant among content files of the same category, such as episodes of a TV show or same-sport recordings, which, even if personalized,
may share common backgrounds and scene objects. 
Hence, we assume that the joint distribution of the library files, denoted by $P_{\mathcal W}$, is not necessarily the product of the file marginal distributions.

Consistent with the existing information-theoretic literature on cache-aided networks ~\cite{maddah14fundamental, maddah13decentralized, ji14average,  ji15order,ji14groupcast}, we assume that the network operates in two phases:
a caching (or placement) phase taking place at network setup, in which caches are populated with content from the library, followed by a delivery phase where the network is used repeatedly in order to satisfy receiver demands. The design of the caching and delivery phases forms what is referred to as a {\em caching scheme}.

The same network setup, but under the assumption of independent library files, was studied in \cite{maddah14fundamental, maddah13decentralized, ji14average, ji15order, ji14groupcast}. These works characterized the order-optimal rate-memory trade-off for both worst-case \cite{maddah14fundamental, maddah13decentralized} and random demand settings \cite{ji14average, ji15order, ji14groupcast}. In this context, each file in the library is treated as an independent piece of information, compressed up to its entropy $F$. During the caching phase, 
parts of the library files are cached at the receivers according to a properly designed {\em caching distribution}. 
The delivery phase consists of computing an {\em index code}, in which the sender compresses the set of requested files into a multicast codeword, only exploring perfect matches (``correlation one") among (parts of) the requested and the cached files, ignoring all other correlations that exist among the different (parts of) files.

The goal of this paper is to investigate the additional gains that can be obtained by exploring the correlations among the library content, in both caching and delivery phases.
An initial step
towards the characterization of the rate-memory region for a correlated library was given in \cite{timo2015rate}.
The authors focus on lossy reconstruction in a scenario with only two receivers and one cache. 
In contrast, we address the more general multiple cache scenario, and focus on correlation-aware lossless reconstruction.
We formulate the problem via information-theoretic tools where the caching phase allows for placement of arbitrary functions of a correlated library 
at the receivers, while the delivery phase becomes equivalent to a source coding problem with distributed side information. 
A correlation-aware scheme then consists of receivers storing content pieces based on their popularity as well as on their correlation with the rest of the library in the caching phase, and receiving compressed versions of the requested files according to the information distributed across the network and their joint statistics during the delivery phase. 
We introduce a lower bound for the rate-memory trade-off as well as an achievable scheme 
that includes a library compression step proceeding correlation-aware versions of the caching and delivery phases previously proposed in \cite{ji14average, ji15order}. 
We provide comparisons with state of art correlation-unaware schemes and an alternative correlation-aware achievable strategy introduced in \cite{ISTC2016}.


The paper is organized as follows. 
Sec. \ref{sec:Problem Formulation} presents the information-theoretic problem formulation. Sec. \ref{sec:Lower bound} provides a lower bound on the fundamental rate-memory trade-off. An achievable scheme is described in Sec. \ref{sec:LibraryCompressed Scheme}, 
its performance is numerically validated in Sec. \ref{sec:Simulations}, and the conclusion follows in Sec. \ref{sec:Conclusion}.

\section{Problem Formulation}~\label{sec:Problem Formulation}

\noindent{\bf Notation:} For ease of exposition, we use $\{A_i\}$ to denote the set of elements $\{A_i:i\in\mathcal I\}$, with $\mathcal I$ being the domain of index $i$. $\FF_2^*$ denotes the set of finite length binary sequences. We use $[x]^+$ to denote $\max\{x,0\}$.

We consider the following information-theoretic formulation of a caching scheme for $n$ receivers and $m$ files: 
\begin{itemize}
\item Initially, a realization of the library $\{W_f\}$ is revealed to the sender.
\item {\textbf{Cache Encoder:}}
At the sender, the cache encoder computes the content to be placed at the receiver caches by using the set of functions
$\{Z_u: \FF_2^{m F} \rightarrow \FF_2^{MF} : u \in \Uc\}$,  such that $Z_u(\{W_f\})$ is the content cached at receiver $u$, and $M$ is the (normalized) cache capacity or alternatively memory.
The cache configuration $\{Z_u\}$ is designed jointly across receivers, 
taking into account global system knowledge such as the number of receivers and their cache sizes, the number of files, their aggregate popularity, and their joint distribution $P_{\mathcal W}$. Computing $\{Z_u\}$ and populating the receiver caches constitutes the caching phase, which is assumed to happen during off-peak hours without consuming actual delivery rate.
\item{\textbf{Multicast Encoder:}}
Once the caches are populated, the network is repeatedly used for different demand realizations. At each use of the network, a random demand vector $\fsf = (\fsf_1,\dots,\fsf_n) \in \mathcal F^{n}$ is revealed to the sender. We assume that $\fsf$ has i.i.d components distributed according to $\qbf$ and we denote by $\fbf =(f_1,\dots,f_n)$ its realization.
The multicast encoder is defined by a fixed-to-variable encoding function $X : \Fc^n \times \FF_2^{mF} \times \FF_2^{nMF} \rightarrow \FF_2^*$, such that $X(\fbf, \{W_f\},\{Z_u\})$ is the transmitted codeword generated according to demand realization $\fbf$, library realization $\{W_f\}$, cache configuration $\{Z_u\}$, and joint file distribution $P_{\mathcal W}$.
\item{\textbf{Multicast Decoders:}}
Each receiver $u\in\mathcal U$ recovers its requested file $W_{f_u}$ using the received multicast codeword and its cached codeword, as $\widehat{W}_{f_u} = \rho_u(\fbf, X,Z_u)$, where $\rho_u :  \Fc^n \times \FF_2^* \times \FF_2^{MF}  \rightarrow \FF_2^{F}$ denotes the decoding function of receiver $u$.
\end{itemize}

The worst-case (over the file library) probability of error of the corresponding caching scheme is defined as
\begin{align} \label{perr}
& P_e^{(F)} = \sup_{\{W_f : f \in \mathcal F\}} \PP \left(\widehat{W}_{f_u}  \neq W_{f_u} \right). \notag
\end{align}
In line with previous work { \cite{ji14average,ji15order,ji14groupcast}}, 
the (average) rate of the overall {caching scheme} is defined as
\begin{equation} \label{average-rate}
R^{(F)} = \sup_{\{ W_f : f \in \mathcal F\}} \; \frac{\EE[J(X)]}{F}.
\end{equation}
where $J(X)$ denotes the length (in bits) of the multicast codeword $X$.
\begin{definition} \label{def:achievable-rate}
A rate-memory pair $(R,M)$ is {\em achievable} if there exists a sequence of caching schemes for cache capacity (memory) $M$ and increasing file size $F$ such that
{$\lim_{F \rightarrow \infty} P_e^{(F)} = 0 \notag$, and $\limsup_{F \rightarrow \infty} R^{(F)} \leq  R.$}
\end{definition}
\begin{definition} \label{def:infimum-rate}
The rate-memory region is the closure of the set of achievable rate-memory pairs $(R,M)$. The rate-memory function $R(M)$ is the infimum of all rates $R$ such that $(R,M)$ is in the rate-memory region for memory $M$.
\end{definition}
In this paper, we find a lower bound and an upper bound on the rate-memory function $R(M)$, given in Theorems \ref{thm:LowerBound} and \ref{thm:uniform compressed library} respectively, and design a caching scheme (i.e., a cache encoder and a multicast encoder/decoder) that results in an achievable rate $R$ close to the lower bound.

\section{Lower bound}~\label{sec:Lower bound}
In this section, we derive a lower bound on the rate-memory function under uniform demand distribution using a cut-set bound argument on the broadcast caching-demand augmented graph \cite{llorca2013network}. 
To this end let $\mathcal D^{(j)}$ denote the set of demands with exactly $j$ distinct requests. 

\begin{theorem}\label{thm:LowerBound}
For the broadcast caching network with $n$ receivers, library size $m$, uniform demand distribution, and joint probability
$P_{\mathcal W}$, 
\begin{eqnarray} \label{eq: converse}
&& \! \! \! \!  \! \! \! \!  R(M) \geq 
\liminf\limits_{F \rightarrow \infty} \max_{\ell \in \left\{1, \cdots, \gamma\right \} } \! \! 
P_\ell \, \frac{ 
\Scale[0.9]{ \left[
 H \! \left ( \!  \left \{ {\sf W}_f:  \, f  \in \mathcal F  \right  \}  \right )  - \ell M F \right]^+ } }{\left \lfloor \frac{m}{\ell} \right \rfloor    F } \notag
\end{eqnarray}
where  $ \gamma = \min\left\{n,m\right\}$, 
$P_\ell$$=$$P( \dbf \in \bigcup_{j\geq \ell} \mathcal D^{(j)}  )$ 
and  $H \! \left ( \!  \left \{ {\sf W}_f:  \, f  \in \mathcal F  \right  \}  \right )$ is the entropy of the entire library.
\end{theorem}
\noindent
\begin{proof}
\noindent 
The proof of Theorem \ref{thm:LowerBound} is given in Appendix \ref{App: Lower Bound}. 
\end{proof}

\section{Library Compressed Cache-Aided Coded Multicast (Comp-CACM)}\label{sec:LibraryCompressed Scheme}
In this section we describe a correlation-aware caching scheme, termed Comp-CACM. 
The proposed Comp-CACM scheme consists of a {\em Correlation-Aware Library Compressor}, a {\em Random Fractional  Cache Encoder} and a {\em Coded Multicast Encoder}.
In this scheme: 
\begin{itemize} 
\item { \em Step.1}: 
The library files are divided into two subsets, {\em I-files} and {\em P-files}. 
\item { \em Step.2}: 
The P-files are conditionally compressed with respect to the I-files. 
\item { \em Step.3}: 
 The compressed files are divided into packets, and a portion of these packets is cached at each receiver according to a caching distribution.
 \item { \em Step.4}: 
During the delivery phase, the receivers requesting an I-file receive a linear combination of the uncached packets of the corresponding I-file, and those requesting a P-file receive a linear combination of the uncached packets of the corresponding P-file and the uncached packets of the associated I-file. 
\end{itemize}   
As described in more detail in the subsequent sections, Steps 1 and 2 are carried out by the Correlation-Aware Library Compressor, Step 3 is done by the Random Fractional Cache Encoder and the Coded Multicast Encoder performs Step 4.
Differently from existing correlation-unaware schemes,
in Comp-CACM, 
the cache encoder operates on a compressed version of the library obtained by the Correlation-Aware Library Compressor. Even though the delivery phase is still based on computing an index code, this phase needs to be properly modified to account for the fact that the content stored in the receiver caches is a compressed version of the library files. 

\subsection{Correlation-Aware Library Compressor}\label{comp}
In Steps 1 and 2, prior to passing the library to the cache encoder, the correlations among the library content are exploited to compress the files according to their joint distribution $P_{\mathcal W}$ as well as their popularity, resulting in a library composed of {\em inter-compressed} files.

\begin{definition}({\bf $\delta$-Correlated Files})
\label{eq:def-detalcor}
For a given threshold $\delta\leq 1$ and file size $F$, we say file ${\sf W}_{f}$ is $\delta$-correlated with file ${\sf W}_{f'}$ if $H({\sf W}_{f},{\sf W}_{f'}) \leq (1+\delta) F$ bits.
\end{definition}



Using the above notion of $\delta$-correlated files, the file library $\mathcal F = \mathcal F_I\cup \mathcal F_P$ is partitioned into {I-files} $\mathcal F_I$, and inter-compressed {P-files} $\mathcal F_P$, as follows: 
\begin{enumerate}
\item Initialize the file set to the original library $\mathcal F$.
\item Compute the {\em aggregate popularity} for each file in the file set as the probability that at least one receiver requests at least one file from its {\em $\delta$-file-ensemble}, defined as the given file and all its $\delta$-correlated files in the file set. 
\item Choose the file with the highest {aggregate popularity}, $W_{f^*}$, 
and assign it to the I-file set $\mathcal F_I \subseteq\mathcal F$. All other files in the $\delta$-file-ensemble of $W_{f^*}$ are inter-compressed using  $W_{f^*}$ as reference and assigned to the P-file set $\mathcal F_P \subseteq\mathcal F$.
\item Update the file set by removing all files contained in the $\delta$-file-ensemble of $W_{f^*}$.
\item Repeat steps  2-4 until the file set is empty.
\end{enumerate}

The partitioning procedure described above results in a compressed library, $\{\overline{W}_f:   f\in\mathcal F\}$, composed of I-files $\mathcal F_I$ 
each with entropy $F$, and P-files $\mathcal F_P$ 
inter-compressed with respect to their corresponding I-files, with conditional entropy less than $\delta F$. The value of 
 $\delta$  will be further optimized as a function of the system parameters ($n,m,M,\qbf,\mathcal P_W$) in order to minimize the (average) rate of the overall proposed {caching scheme} given in (\ref{average-rate}). 

The following example illustrates the compression of the library for a simple system setup. 
\begin{example}\label{ex: Compressed Library}
Consider a library consisting of $m=4$ uniformly popular files $\{W_1,W_2,W_3,W_4\}$ each with entropy $F$ bits. We assume that the file pairs $\{W_1 ,W_2\}$ and $\{W_3, W_4\}$ are independent, while correlations exist between $W_1$ and $W_2$, and between $W_3$ and $W_4$. Specifically, $H({\sf W}_1|{\sf W}_2)=H({\sf W}_2|{\sf W}_1)=F/4$ and $H({\sf W}_3|{\sf W}_4)=H({\sf W}_4|{\sf W}_3)=F/4$. 
Since files are equally popular and each file is correlated with one other file, each file has the same aggregate popularity. We choose file $W_2$ as the first I-file and assign its $\delta$-correlated file, $W_1$, to the P-file set. 
The file set is then updated to $\{W_3,W_4\}$. Again, either file can be assigned to the I-file set as they have the same aggregate popularity. We choose $W_4$ as the second I-file and assign $W_3$ to the P-file set. 
This results in the compressed library $\{\overline{W}_f\}=\{\overline W_1,W_2,\overline W_3,W_4\}$ with $\mathcal F_I=\{2,4\}$ and $\mathcal F_P=\{1,3\}$. $\overline W_1 (\overline W_3)$ denotes P-file $W_1(W_3)$ inter-compressed with respect to I-file $W_2(W_4)$. 
\end{example}

\subsection{Random Fractional Cache Encoder} 

After compressing the library, in Step 3, Comp-CACM applies a slightly modified version of the random fractional cache encoder 
introduced in \cite{ji15order,ji14average} to the compressed library. 

Specifically, given the compressed library $\{\overline{W}_f\}$, the cache encoder divides each file into equal-size packets, such that each file $ \overline{W}_f$ is divided into $H(\overline{W}_f)/b\in \mathbb N$ packets of length $b$ bits. The packet length $b$ is chosen such that all files in the compressed library (of possibly different lengths) are divided into an integer number of packets. The sender then populates the caches by randomly selecting and storing $p_fM$ fraction of the packets of each file $\overline W_f$,$f\in{\mathcal F}$ at each receiver, where $\pbf = ( p_{1},\dots,p_{m} )$, referred to as the {\em caching distribution}, is a vector with elements $0\leq p_{f}\leq 1/M$, satisfying the cache constraint as $\sum_{f \in \mathcal F_I} p_{f}+\delta\sum_{f \in \mathcal F_P} p_{f} = 1 $.\footnote{As in \cite{ji15order}, the caching distribution defines the number of packets cached from each file; however, since here I-files and inter-compressed P-files have different lengths and are therefore divided into a different number of packets, the cache constraint differs from that of \cite{ji15order}.} 
 The caching distribution is optimally designed to minimize the rate of the corresponding correlation-aware delivery scheme as expressed in (\ref{average-rate}) while taking into account global system parameters ($n,m,M,\qbf,P_{\mathcal W}$). The rate expression in (\ref{average-rate}) is normalized to the file size, and is independent of packet length $b$ when $F\rightarrow\infty$. In the following, we denote the $i^{th}$ packet of file $W_f, f\in\mathcal F$ by $W_{fi}$.


\begin{example}\label{ex: Comp RAP}
In Example \ref{ex: Compressed Library} further assume that the sender is connected to $n=2$ receivers $\{u_1, u_2\}$ each with a cache size $M=1$. For the compressed library $\{\overline W_1,W_2,\overline W_3,W_4\}$ described in Example \ref{ex: Compressed Library}, we assume packet length $b=F/4$ and a caching distribution such that $p_{2}=p_{4}=1/4$ and $p_{1}=p_{3}=1$. This corresponds to splitting files $W_2$ and $W_4$ into $4$ packets, each with entropy $F/4$, and randomly storing $1$ packet from each file at each receiver. Files  $\overline W_1$ and $\overline W_3$ comprise one packet each and are both cached at the receivers. We assume that packets $\{ \overline W_{1},W_{21}, \overline W_{3},W_{41}\}$  are cached at $u_1$ and  $\{\overline W_{1},W_{22},\overline W_{3},W_{42}\}$ are cached at $u_2$, as shown in Fig \ref{fig:Examples}.
\end{example}

\begin{figure}
  \centering
  \includegraphics[width=3.5in]{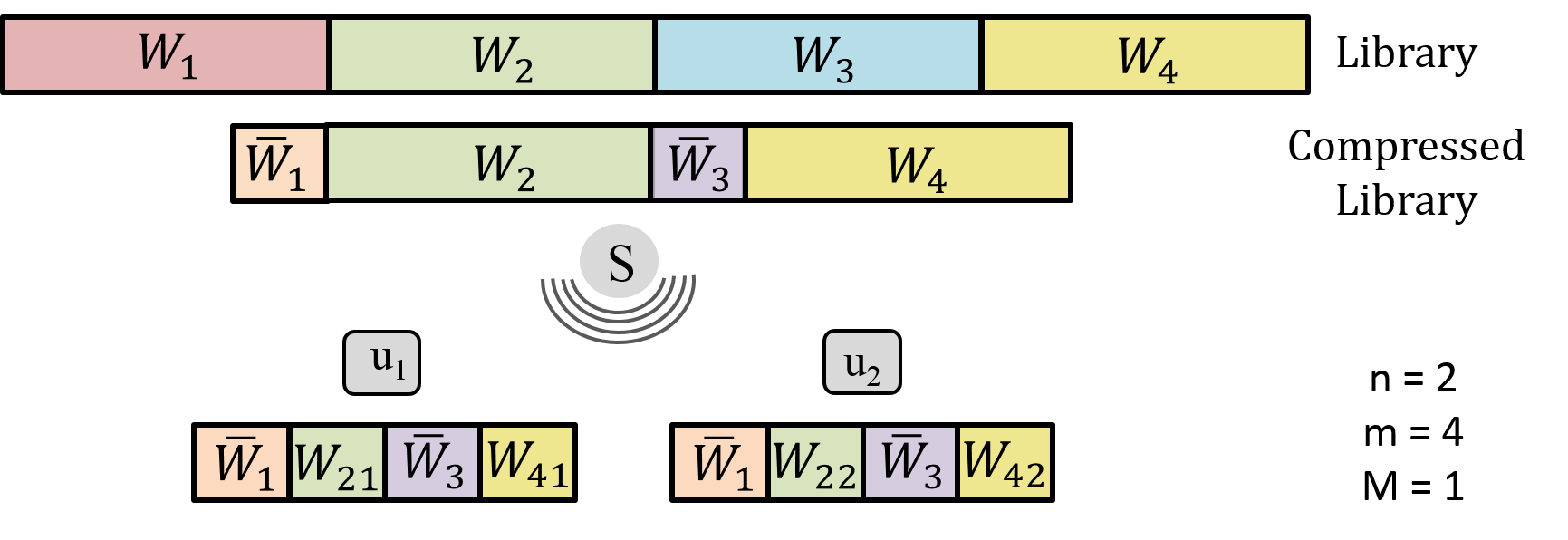}
  \caption{Network setup and cache configuration of Examples \ref{ex: Compressed Library}-\ref{ex: Comp GCC}.
} 
  \label{fig:Examples}
\end{figure}

\subsection{Coded Multicast Encoder} 

The goal of the multicast encoder of Step 4 is to generate a coded multicast codeword that allows each receiver to losslessly reconstruct its requested file, given that the cached content and the resulting conflict graph is composed of inter-compressed file pieces. 
Each receiver recovers its demand using the multicast codeword and its cache content. The multicast codeword can be
a linear combination of the  packets needed by each receiver for demand recovery. The objective is to find the minimum multicast transmissions (termed as the multicast rate) given a set of demands and the cache content of all receivers. 
The problem can be formulated as the coloring of the corresponding conflict graph. However, differently from the classical setting described in \cite{ji14average, ji15order, ji14groupcast}, where each requested packet corresponds to a vertex in the conflict graph, here the vertices are  composed of: 1) for each receiver requesting an I-file, the uncached packets of the corresponding I-file, and 2) for each receiver requesting a P-file, the uncached packets of the corresponding P-file and the uncached packets of the associated I-file \cite{Journal}.  As in the classical setting, a directed edge between two vertices of the graph means that one receiver's desired packet interferes with some other receiver's desired packet (i.e is not cached at that receiver).


\begin{example}\label{ex: Comp GCC}
For the cache configuration given in Example \ref{ex: Comp RAP}, and demand realization $\fbf = (1, 2)$, the sender multicasts $W_{21} \oplus W_{22} , W_{23} , W_{24}$, which corresponds to a total rate of $R = 0.75$. The transmitted codeword enables the reconstruction of file $W_{2}$ at $u_1$, and with $\overline W_{1}$ being in its cache, receiver $u_1$ is able to losslessly reconstruct $W_1$. Receiver $u_2$ is able to reconstruct $W_2$ using the received packets. 
However, a correlation-unaware scheme (e.g., \cite{maddah13decentralized,ji14average}) would first compress the files separately and then cache $1/4^{th}$ of each file at each receiver, which would result in a total rate $R=1.25$ regardless of the demand realization \cite{maddah13decentralized,ji14average}.
\end{example}

\subsection{Rate Upper Bound for Comp-CACM}\label{subsec:library compressed rate}
In this section, we upper bound the rate-memory function given in Definition \ref{def:infimum-rate} by evaluating the performance of the Comp-CACM under the uniform demand scenario.\footnote{Closed-form expressions for the achievable rate of Comp-CACM for the general setting described in Sec. \ref{sec:Introduction} are given in \cite{Journal}.} 
In stating the theorem, we assume that the library file correlation is such that for a given $\delta$, there are at least $\kappa$ files in the $\delta$-file-ensemble of each file $f\in\mathcal F$. In other words, $\kappa$, which is a function of $\delta$ and of the joint library distribution $P_{\mathcal W}$, is the largest value such that for each file $f\in\mathcal F$, there are at least $\kappa$ files $\delta$-correlated with $f$.
 
\begin{theorem}\label{thm:uniform compressed library}
Consider a broadcast caching network with $n$ receivers, library size $m$, and uniform demand distribution. Then, the rate-memory function, $R(M)$, is upper bounded as
\begin{flalign}
&R(M) \leq  \inf_{\delta} \left \{\min\{\psi(\delta,p_I^*,p_P^*,M), \bar m\} \right  \},& \label{upperRstar}\\
\text{where}\notag\\
& \psi(\delta,p_I,p_P,M)  = \sum_{\ell=1}^{n} \binom{n}{\ell} \Big(\lambda_{\ell}(p_I) + \delta \lambda_{\ell}(p_P)\Big),  \notag\\
& \{p_I^*,p_P^*\} = \argmin_{p_I,p_P}\{\psi(\delta,p_I,p_P,M), \bar m\} \notag\\
&  \qquad\qquad\quad \mbox{\rm{s.t.}}  \quad   p_I + \delta(\kappa-1) p_P = \frac{\kappa}{m}  \label{kappa}\\
& \bar m = m \Big( 1- \Big(1-\frac{1}{m}\Big)^{n} \Big),\notag\\ 
\text{with} \notag\\
& \lambda_{\ell}(p) \triangleq {(pM)}^{(\ell-1)}(1-{pM})^{(n-\ell+1)}.  \notag
\end{flalign}
with $\kappa$ in (\ref{kappa}) being a function of $\delta$ and $P_{\mathcal W}$.
\end{theorem}
\begin{proof}
The proof of Theorem \ref{thm:uniform compressed library} is given in \cite{Journal}.
\end{proof}

Note that for a given $\delta$, the argument of the infimum on the right-hand side of \eqref{upperRstar} represents an upper bound on the Comp-CACM rate achieved with the caching distribution optimized for uniformly popular files, given by  $ p_{f} = p^*_I$, $\forall f \in \mathcal F_I$, and  $ p_{f} = p^*_P $, $\forall f \in \mathcal F_P$.
As mentioned in the previous section, the caching distribution of Comp-CACM is designed to minimize the rate of the corresponding correlation-aware delivery scheme. 
Consequently, the caching distribution  of Comp-CACM not only depends on the demand distribution, as in 
correlation-unaware schemes, 
but also depends on the correlations among the library content. Hence, under a uniform demand distribution, differently from correlation-unaware schemes, 
where all files are uniformly cached due to their uniform popularity, the Comp-CACM caching distribution assigns a larger portion of the memory to storing I-files compared to storing P-files, 
due to I-files having a higher {\em aggregate popularity} than P-files (see Sec. \ref{comp}). 
\section{Numerical Results and Discussions}~\label{sec:Simulations}
We numerically compare the rates achieved by the proposed correlation-aware scheme Comp-CACM, with respect to the following schemes: Local Caching with Unicast Delivery (LC/U), Local Caching with Naive Multicast Delivery (LC-NM), Random Popularity based (RAP) Caching with Coded Multicast Delivery (RAP/CM) \cite{ji15order}, all of which are state-of-the-art caching schemes that do not consider and exploit content correlation. We also compare the rates with an alternative correlation-aware scheme, termed CA-RAP/CM 
introduced in \cite{ISTC2016,Journal}, in which receivers store content pieces from the original (not inter-compressed) library based on their popularity and their correlation with the rest of the file library during the caching phase, and receive compressed versions of the requested files according to the information distributed across the network and their joint statistics during the delivery phase. Even though the caching of content in CA-RAP/CM is not done as efficiently as in Comp-CACM, the delivery is more efficient. 
 All of the schemes mentioned above are compared with the rate-memory function lower bound given in Theorem \ref{thm:LowerBound}. 

We consider a broadcast caching network with $n=10$ receivers requesting files according to a uniform demand distribution from a library composed of $m=100$ files, in which there are $\kappa$ files in the $\delta$-ensemble of each file. 

Fig. \ref{fig:alpha 0} displays the rate-memory trade-off as the memory size varies from $0$ to $100$ files for $\delta = 0.2$ and $\kappa = 2$. 
The figure plots the expected achievable rate $R$, in number of transmissions normalized to the file size, versus memory size (cache capacity) $M$, normalized to the file size. As expected, the correlation-aware schemes outperform schemes that are oblivious to the file correlations. Comp-CACM  achieves a $3.5\times$ reduction in the expected rate compared to LC-U and a $1.5\times$ reduction compared to RAP/CM for $M=20$. When comparing Comp-CACM and CA-RAP/CM, for small memory size $M$, CA-RAP/CM outperforms Comp-CACM, while the opposite is observed for large $M$. This indicates that joint compression of the file library has an especially relevant effect at large memory size, since each receiver is able to cache most of the compressed library. 
In fact, Comp-CACM achieves zero rate at around $M=60$, indicating that every receiver is able to cache the entire library in its compressed form. For small cache size $M$, however, the extra bits that Comp-CACM is forced to request in the delivery phase due to the initial library compression undermines its benefit with respect to  CA-RAP/CM.

\begin{figure}
  \centering
  \includegraphics[width=2.9in]{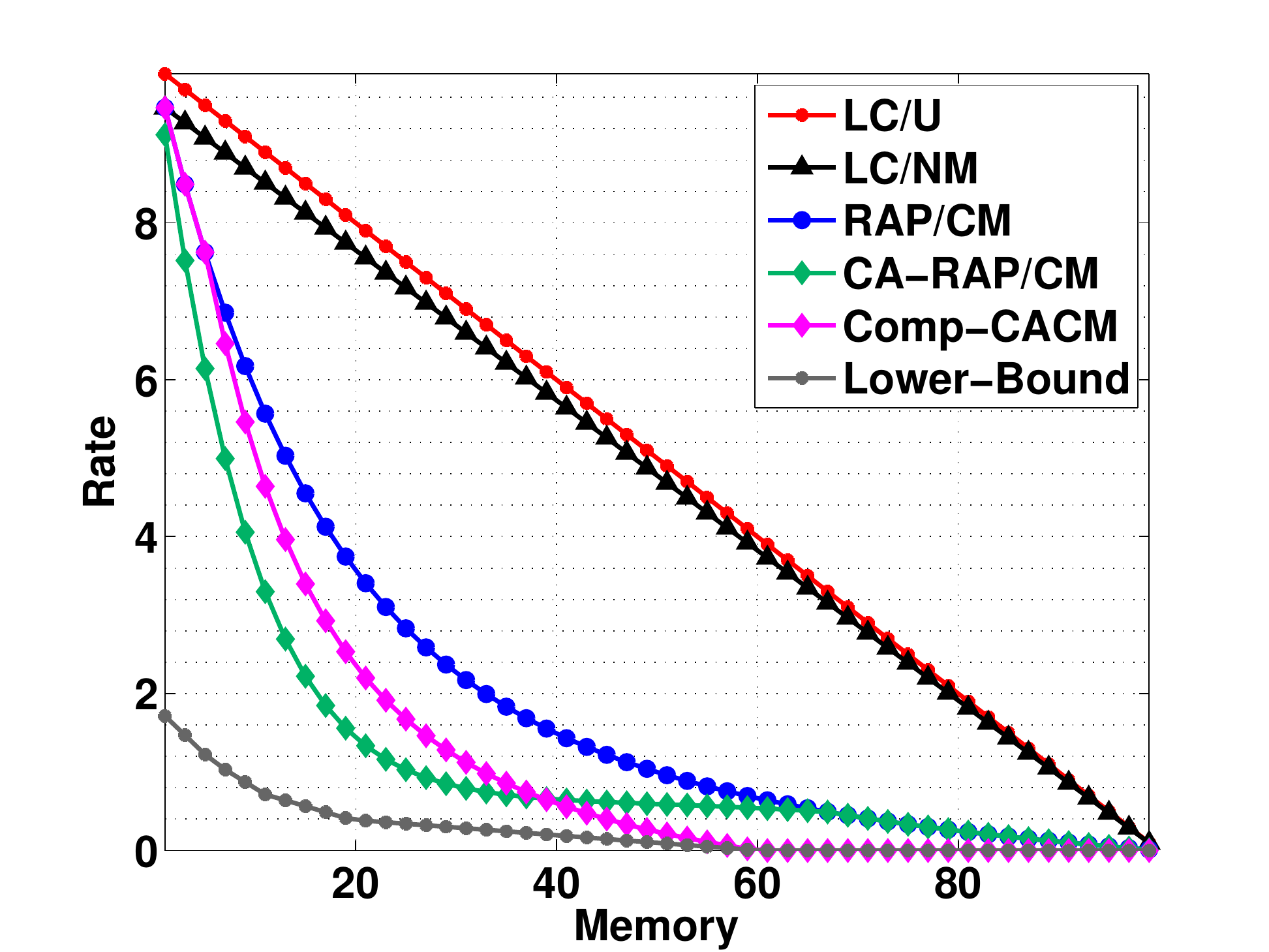}
  \caption{Comparison of Comp-CACM and the lower bound in Theorem \ref{thm:LowerBound} with other caching schemes in a network with uniform demand distribution, $n = 10$, $m = 100$, $\delta = 0.2$, and $\kappa = 2$.}
  \label{fig:alpha 0}
\end{figure}

\section{Conclusion}~\label{sec:Conclusion}
In this paper we have formulated the caching problem using information-theoretic tools and have provided a lower bound on the fundamental rate-memory trade-off. We have proposed Comp-CACM, an achievable correlation-aware scheme, in which, prior to storing the content files, the library is jointly compressed while taking into account the popularity of the files as well as their joint distribution. This leads to caching of the more relevant content during the caching phase, which, with the transmissions during the delivery phase that ensure lossless reconstruction of the demand, results in the more efficient delivery of content. We have then compared the performance of Comp-CACM with CA/RAP-CM, another correlation-aware scheme, 
as well as with the conventional caching schemes. 

Part of our ongoing work focuses on i) designing improved achievable schemes that incorporate the beneficial compression aspects of Comp-CACM into CA-RAP/CM, such as jointly compressing the cache configuration at each receiver, and ii) studying the effect of relevant system parameters that characterize the correlated library such as  $\delta$ and $\kappa$. 

\begin{appendices}

\section{Proof of Theorem \ref{thm:LowerBound}}\label{App: Lower Bound}
Denoting by $\mathcal D^{(j)}$ the set of all demands that contain at least $j$ distinct requests,
it is immediate to prove that  
\begin{eqnarray} 
R(M) & \geq&  P \left (  \dbf \in \mathcal  \bigcup_{j\geq  \ell}\mathcal D^{(j)}  \right ) R_\ell(M)
\label{eq:rate0}
\end{eqnarray} 
where $R_\ell(M)$ is the average rate needed to satisfy the user demands  containing exactly $\ell$ distinct requests.
Due to the uniform popularity assumption, we have that   
\begin{eqnarray}
R_\ell(M) =\frac{1}{|\mathcal D^{(\ell)} |} \sum_{\dbf \in \mathcal D^{(\ell)} }  R_\ell(M, \dbf )
\label{eq:rate1}
\end{eqnarray} 
where $R_\ell(M, \dbf )$ denotes the rate needed to satisfy user demand $\dbf$. Let $\{\dbf_1, \ldots, \dbf_{\lfloor \frac{ m}{\ell}\rfloor}\}$ denote a set of $\lfloor \frac{ m}{\ell}\rfloor$ demands in $\mathcal D^{(\ell)}$
such that the union of all their requested files is equal to the entire library. 
Furthermore, 
by the information-theoretic cut-set bound we have that 
\begin{eqnarray}
\sum_{\dbf \in  \left \{\dbf_1, \ldots, \dbf_{\lfloor \frac{ m}{\ell}\rfloor} \right \} }  R_\ell(M, \dbf )F  + \ell M F \geq 
H \! \left ( \!  \left \{ {\sf W}_f:  \,   f \in \mathcal F  \right  \}  \right ) 
\label{eq:rate2}
\end{eqnarray} 
Replacing \eqref{eq:rate2}  in \eqref{eq:rate0}, Theorem \ref{thm:LowerBound} follows. 
 
\bibliographystyle{IEEEtran}
\bibliography{ITW_arxiv_V1}

\end{appendices}

\end{document}